\documentclass{Manuscript_Template_ICCA9_LaTex}
\usepackage{graphicx}
\usepackage{amsmath}
\usepackage{amssymb}
\usepackage{amsthm}
\usepackage{dsfont}
\usepackage{url}

\newcommand{\cs}{\mathfrak{S}}
\newcommand{\cf}{\mathcal{F}}

\newcommand{\cm}{\mathcal{M}}
\newcommand{\ct}{\mathcal{T}}
\newcommand{\ck}{\mathfrak{k}}
\newcommand{\ch}{\mathfrak{h}}

\newcommand{\cd}{\mathcal{D}}
\newcommand{\bcd}{\bar{\mathcal{D}}}
\newcommand{\hcd}{\hat{\mathcal{D}}}

\newcommand{\BC}{\mathbb{C}}
\newcommand{\BR}{\mathbb{R}}

\newcommand{\be}{\begin{equation}}
\newcommand{\ee}{\end{equation}}

\newcommand{\I}{\mathrm{i}} 
\newtheorem{proposition}{Proposition}

\newtheorem{remark}{Remark}
\newtheorem{lemma}{Lemma}
\renewcommand\Im{\operatorname{\mathfrak{Im}}}
\renewcommand\Re{\operatorname{\mathfrak{Re}}}
\title{GRAVITATION ON A HOMOGENEOUS DOMAIN}
\author{Arkadiusz Jadczyk}
\address{Center CAIROS, Institut de Math\'{e}matiques de Toulouse\\
Universit\'{e} Paul Sabatier, 118 Rout\'{e} de Narbonne 31062, Toulouse Cedex 9, France\\
\selectfont \normalfont
E-mail: arkadiusz.jadczyk@cict.fr}
\abstract{Among all plastic deformations of the gravitational Lorentz vacuum \cite{wr1} a particular role is being played by conformal deformations. These are conveniently described by using the homogeneous space for the conformal group $SU(2,2)/S(U(2)\times U(2))$ and its Shilov boundary - the compactified Minkowski space $\tilde{M}$ \cite{aj1}.  In this paper we review the geometrical structure involved in such a description. In particular we demonstrate that  coherent states on the homogeneous K\"{a}hler domain give rise to Einstein-like plastic conformal deformations when extended to $\tilde{M}$ \cite{aj2}.}
\begin{document}
\section{Introduction}
William Kingdon Clifford speculated \cite[p. 22]{tucker} that the curvature of space is responsible for all motions of matter and fields - the idea that has been taken over by Albert Einstein in his theory of gravitation, through with the extra assumption of the weak equivalence and general covariance principles. P. A. M. Dirac, originally impressed by General Relativity Theory, later on had his doubts about the validity of general covariance, when the lessons of quantum theory are taken into account. He tried to revive and reformulate the old idea of aether \cite{dirac}. The idea that an alternative to Einstein's gravity is needed in order to reconciliate, somehow, classical geometry with quantum theory is, at least, an interesting one.\footnote{Another approach is that of noncommutative geometry \cite{connes}} In the present paper we study gravitational fields that are space/time imprints of coherent quantum states on a homogeneous complex domain for the conformal group. We start with the simplest toy case of the Poincar\'{e} disk that is a homogeneous space for the group $SU(1,1).$ It's Shilov's boundary - cf. \cite{coja90} and references threre) is just the unit circle, which plays the role of the compactified Minkowski space in this case. Since the circle is one--dimensional, Riemannian metrics on such a space are easy to describe - they are represented by  positive functions on the circle. Using Cayley's transform the circle (minus one point) is mapped onto $\BR.$ We study a particular family of quadratic functions over $\BR$ (a special family of parabolas) and show that they are generated by coherent quantum states on the unit disk. Then we move to the case of interest, namely the complex homogeneous bounded domain $D=SU(2,2)/S(U(2)\times U(2))$ and study a particular class (transitive under the action of $SU(2,2)$) of coherent states on $D.$ In this case Shilov's boundary of $D$ is the compactified Minkowski space, and we show that the imprints of these stats on the boundary can be interpreted as gravitational fields in the conformal class of the Minkowski metric. In fact, we show that what we get is a family of de Sitter type metrics.
\section{The space $H^2$}
Consider the following 2--parameter family of parabolas:
\be y_{q,p}(x)=\frac{1}{2}\left(\frac{(x-q)^2}{p}+p\right),\quad q\in \BR,\, p>0.\label{eq:para}\ee For each of the parabolas its focus is at $(q,p)$ and the distance between the minimum and the focus is the same as between the minimum and $p=0$ axis.
Each of these parabolas represents a vector field on $\BR:$
\be e(x;q,p)=y_{q,p}(x).\label{eq:vf}\ee
Note: Because we are in one dimension, we will suppress covariant and contravariant indices. Let $\omega(x;q,p)$ be the 1--form dual to $e(x;q,p):$
$$\omega(x;q,p)=1/e(x;q,p),$$ then $\omega$ defines the metric \be g=\omega\otimes\omega=\frac{4}{\left(p+\frac{(x-q)^2}{p}\right)^2},\label{eq:g}\ee whose volume form is $\omega.$ Its inverse is
$g^{-1} = \frac{1}{4} \left(p+\frac{(x-q)^2}{p}\right)^2,$ Let $\ch,\ck$ be two vectors tangent to the space of (covariant) metrics. We define their scalar product at $g$ by the standard formula \cite{michor1}:
$(\ch,\ck)_{g} = \int_\BR g^{-1}\, \ch\, g^{-1}\,\ck\, \omega\, dx,$
where we omit the trace.\\
From Eq. (\ref{eq:g}) we have
\begin{eqnarray*}\ch = dg &=& \frac{\partial\, g (x;q,p)}{\partial\, q}dq\,+\frac{\partial\, g (x;q,p)}{\partial\, p}dp=\\ &=&
\frac{8 p \left(2  p (x-q)dq+ \left(q^2-p^2-2 q x+x^2\right)\right)dp}{\left(q^2+p^2-2 q x+x^2\right)^3}.\end{eqnarray*}
We can now calculate the induced quadratic form $(\ch,\ch)_{g} = \int_\BR g^{-2}\, \ch^2\,\omega\, dx.$ The integrand is
$(8 (2 p (x-q)dq + (q^2 - p^2 - 2 q x + x^2)dp)^2)/(p (q^2 +
   p^2 - 2 q x + x^2)^3),$ with the primitive function
\begin{eqnarray*}-\frac{2}{p} \left(\frac{2 (x-q) (-dp (q+p-x)+dq (-q+p+x)) (dq (q+p-x)+dp (-q+p+x))}{\left(p^2+(x-q)^2\right)^2}\right)+\\ \frac{2}{p} \left(\frac{2 \left(dq^2+dp^2\right) \text{ArcTan}\left[\frac{x-q}{p}\right]}{p}\right),\end{eqnarray*}
   and the integral from $-\infty$ to $+\infty$ gives (only the second term contributes)
$(\ch,\ch)_{g} = \frac{4 \left(dq^2+dp^2\right) \pi }{p^2},$ which is, up to a constant proportionality factor, the standard Bolyai-Lobachevsky hyperbolic metric on the upper $(q,p)$ half--plane.
\subsection{The group $SL(2,\BR)$}
Let us first recall some classical facts. We denote by $\mathds{H}^2$ the upper half--plane \be\mathds{H}^2=\{z=q+\I p:q\in \BR,\, p>0\}.\ee  The group $SL(2,\BR)$ of $2\times 2$ real  matrices of determinant $1$ acts on $\mathds{H}^2$ by fractional--linear transformations. For a matrix $A=\left(\begin{smallmatrix}a&b\\c&d\end{smallmatrix}\right)$ we denote by $\sigma_A$ the transformation
$\sigma_A:\,z\mapsto\sigma_g(z)=\frac{az+b}{cz+d}.$ Then, with $A\in SL(2,\BR)$ and $z=q+\I p$ we have
\be \Re\left(\sigma_A(z)\right) =\frac{b d-q+2 a d q+a c \left(q^2+p^2\right)}{d^2+2 c d q+c^2 \left(q^2+p^2\right)},\label{eq:resigma}\ee
\be \Im\left(\sigma_A(z)\right)=\frac{p}{d^2+2 c d q+c^2 \left(q^2+p^2\right)}.\label{eq:imsigma}\ee In particular $\Im\left(\sigma_A(z)\right)>0$ if $p>0.$ The Jacobian matrix $J_A(q,p)=\frac{\partial\,\sigma_A(q+\I p)}{\partial\,(q,p)}$ implementing the
tangent map $d\sigma_A$ at $z$ is given by:
$$J_A(q,p) =
m_A(q,p)\begin{pmatrix}\left(d^2+2 c d q+c^2 \left(q^2-p^2\right)\right)&2 c(d+c q) p\\
-2 c(d+c q) p&\left(d^2+2 c d q+c^2 \left(q^2-p^2\right)\right)\end{pmatrix},$$
where $m_A(q,p)=\frac{ad-bc}{\left(d^2+2 c d q+c^2 \left(q^2+p^2\right)\right)^2}.$
Let \be G(q,p)=\frac{1}{p^2}\left(\begin{smallmatrix}1&0\\0&1\end{smallmatrix}\right)\label{eq:G}\ee be the standard hyperbolic metric on $\mathds{H}^2.$ Then, for $A\in SL(2,\BR),$ by a straightforward calculation,
\be {}^tJ(q,p)G(\sigma_A(z))J(q,p)= \frac{1}{\Im(\sigma_A(z))}{}^tJ(q,p)J(q,p)=\frac{1}{p^2}I = G(q,p),\ee
so that $G$ is invariant under $SL(2,\BR)$ transformations.\vskip0.5cm

The group action of $SL(2,\BR)$ extends to the real line (except for a possible singular point if $cx+d=0$), which we will also denote by the letter $\sigma.$
\begin{proposition}The system of vector fields $e(x;q,p)$ is covariant under the action of $SL(2,\BR):$
$ e(\sigma_A(x);\sigma_A(q,p))=d\sigma_A(e(x;q,p)).$
\end{proposition}
\begin{proof}
By substituting $\sigma_A(q+\I p)$ and $\sigma_A(x)=\frac{b (d+c x)+a \left(d x+c x^2\right)}{d^2+2 c d x+c^2 x^2}$ from Eqs. (\ref{eq:resigma},\ref{eq:imsigma}) into Eq. (\ref{eq:para}) we obtain that
$e(\sigma_A(x);\sigma_A(q,p))=\frac{\left(q^2+p^2-2 q x+x^2\right)}{2 p (d+c x)^2}.$ On the other hand we have for $d\,\sigma_A(x)$ the explicit formula:
$\frac{\partial\,\sigma_A(x)}{\partial\, x}=\frac{1}{(d+c x)^2},$ and therefore
$\frac{\partial\,\sigma_A(x)}{\partial\, x}\times e(x,q,p)=\frac{\left(p+\frac{(x-q)^2}{p}\right)}{2 (d+c x)^2}.$
\end{proof}
\subsection{Poncar\'{e} disk $D.$}
The Cayley transform $z\mapsto w(z),$ with
\be w(z) = (z-\I)/(z+\I),\, z=q+\I\,p\label{eq:cayley}\ee
maps the upper half--plane onto the unit disk $D$ in the complex plane. Its inverse is given by \be z(w)=\I(1+w)/(1-w),\label{eq:icayley}\ee with
$p=\Re(z(w))=-2 y/(1-2 x+x^2+y^2),$
$ q=\Im(z(w))=(1-x^2-y^2)/(1-2 x+x^2+y^2).$ Writing $w=x+\I y,$ the tangent map to $w\mapsto z(w)$ is given by the matrix
$$J_c(x,y)=\frac{2}{1-2 x+x^2+y^2}\left(\begin{smallmatrix}2 (-1+x) y&\left(-1+2 x-x^2+y^2\right)\\
\left(1-2 x+x^2-y^2\right)&2 (-1+x) y\end{smallmatrix}\right).$$ Then,  denoting by $I$ the $2\times 2$ identity matrix, by a straightforward calculation we have
$${}^tJ_c(x,y)G(z(w))J_c(x,y)=4I/(1-x^2-y^2)^2,$$ which defines the induced metric on $D.$ The Cayley transform intertwines the fractional--linear transformations by $SL(2,\BR)$ on $\mathds{H}^2$ and fractional--linear transformations by $SU(1,1)$ on $D.$ The connection between the two groups is given by the matrix \cite{habermann1} $\gamma_c$ in $SL(2,\BC):$
$\gamma_c=\frac{1-\I }{2}\left(\begin{smallmatrix}1&-\I \\1&\I \end{smallmatrix}\right),$ with
$\gamma_c^{-1}=\frac{1-\I }{2}\left(\begin{smallmatrix}\I & \I \\-1&1\end{smallmatrix}\right).$
We have
$A\in SU(1,1)$ if and only if $\gamma_c^{-1}A\gamma_c\in SL(2,\BR ).$
The hyperbolic metric $\frac{1}{p^2}\,I$ of $H^2$ is then mapped onto the metric $4I/1-x^2-y^2)^2$ of the Poincar\'{e} disk.

The inverse Cayley inverse transform (cf. Eq.(\ref{eq:icayley})) the unit circle - the boundary of $D$ - to the real line $p=0,$ except for one singular point. Parametrizing the unit circle by $w(t)=\exp(\I t),$ we have
$ z(w(t))=\i \frac{1+\exp(\I t)}{1-\exp(\I t)}=-\cot(t/2),$ with the derivative:
$ \frac{d}{dt}z(w(t))=\frac{1}{2}\csc^2 (t/2).$ The family of metrics $g(q,p;x)$ on $\BR$ is pulled back on the unit circle parametrized by $t$ with the map $t\mapsto z(w(t))$ to give the following metric on the circle:
\be g(\xi;z)=(\frac{d}{dt}z(w(t)))^2 g(q,p;z(w(t)))=\frac{1-\xi\bar{\xi}}{|1-z\bar{\xi}|^2},\label{eq:cmetric}\ee
where we changed the parametrization from $(q,p)\in H^2$ to $\xi=w(q+\I p)$ on $D$ and used $ z=\exp (\I t ).$
The $SL(2,\BR )$--invariant metric $G_H$ (cf. Eq. (\ref{eq:G})) on $H^2$ is pulled back through the inverse Cayley transform $w\mapsto z(w)$ and induces the standard (cf. ref. \cite{canon}) $SU(1,1)$-invariant metric on $D:$
\be ds^2=\frac{4(dx^2+dy^2)}{1-x^2-y^2}.\label{eq:dmetric}\ee
\section{Coherent states on the Poincar\'{e} disk}
In his paper `General Concept of Quantization' \cite{berezin}, F. A. Berezin described, in particular,  quantization on the Poincar\'{e} disk (cf. also \cite{perelomov, antoine}). Here, following \cite[p. 57]{antoine} we will take a small variation of his method as explained below. Berezin starts with the Hilbert space $\mathcal{F}_h$ of analytic functions on $D$ with the scalar product
\be (f,g)=\left(\frac{1}{h}-1\right)\int f(z)\bar{g}(z)(1-z\bar{z})^{\frac{1}{h}}d\mu(z,\bar{z}),\ee
where $d\mu(z,\bar{z})=\frac{1}{2\pi\I }\frac{dz\wedge d\bar{z}}{(1-z\bar{z})^2}$ is the $SU(1,1)$ invariant measure on $D.$ We will take $h=1/2,$ so that the scalar product can be written as
\be (f,g)=\frac{1}{2\pi \I}\int f(z)\bar{g}(z)\,dz\wedge d\bar{z}.\ee
The important role in Berezin's quantization scheme is being played by the family of `coherent states'. To this end we follow \cite{antoine} and introduce the Hilbert space $\mathcal{F}$ of functions square integrable with respect to the invariant measure $d\mu(z,\bar{z}).$ To each point $v\in D$ there is associated a particular vector $\eta_v$ in this Hilbert space given by (using our conventions (cf. also \cite[Eq. (4.99)]{antoine}):
$ \eta_v(z)=\frac{1-v\bar{v}}{(1-z\bar{v})^2}.$
By comparing with Eq. (\ref{eq:cmetric}) we see that on the boundary of $D$ the absolute values of the coherent states $\eta_v$ coincide with the metrics $g(\xi;z).$
\section{Densities for the group $SU(2,2)$}
The group $SU(2,2)$ consists of $4\times4$ matrices
$M=\left(\begin{smallmatrix}A&B\\C&D\end{smallmatrix}\right),$
where $A,B,C,D$ are $2\times2$ complex matrices, satisfying
$M^*GM=G,\quad \det(M)=1,$
$G=\left(\begin{smallmatrix}E&0\\0&-E\end{smallmatrix}\right),$
where ${}^*$ denotes the Hermitian conjugate, and $E$ is the $2\times2$ unit matrix. The inverse $M^{-1}=GM^*G$ is then easily seen to be given by
\be
M^{-1}=\left(\begin{smallmatrix}A^*&-C^*\\-B^*&D^*\end{smallmatrix}\right).\label{eq:inv}\ee
The condition $M^*GM=G,$ when written in terms of $2\times 2$ matrices reads $A^*A-C^*=E,$ $D^*D-B^*B=E,$ $A^*B-C^*D=0,$ or, equivalently, as $MGM^*=G,$ i.e.:
\be AA^*-BB^*=E,\,
DD^*-CC^*=E,\,
AC^*-BD^*=0.
\label{eq:su22}\ee
It follows automatically from these conditions that for the operator norms we have $||A||\geq 1,$ $||D||\geq 1,$ and therefore $A$ and $D$ are invertible. In particular we may apply the following general formula (cf. e.g. \cite{meyer}) for the determinant of the block matrices:
\begin{eqnarray} \det(M)&=&\det(A)\,\det(D-CA^{-1}B)\\
&=&\det(D)\det(A-BD^{-1}C).
\label{eq:dets}\end{eqnarray}
Let $\bcd$ (resp. $\cd$) be the set of all $2\times 2$ complex matrices $Z$ satisfying $Z^*Z\leq 1,$ (resp. $Z^*Z<1$) or, equivalently - invoking the polar decomposition theorem, $ZZ^*\leq 1$ (resp. $ZZ^*<1$). The group $SU(2,2)$ acts on $\bcd$ by linear fractional transformations:
\be M:\,Z\mapsto Z'=(AZ+B)(CZ+D)^{-1},\label{eq:action}\ee
with $CZ+D$ being automatically invertible for $Z\in \bcd.$ The action of $SU(2,2)$ on $\cd$ is transitive -- cf. e.g. \cite{ruhl,jakodzi}.  \footnote{This action is not effective. The kernel of this action is nontrivial ($=Z_4$), and consists of four $4\times 4$ matrices $\{I,-I,iI,-iI\}.$}
\footnote{The action (\ref{eq:action}) can be interpreted in two ways: either as an active transformation of $\cd$ or as a passive change of complex coordinates in $\cd.$}

It follows from Eqs. (\ref{eq:su22}) and (\ref{eq:action}) that
\be E-Z_1'^*Z_2'=(CZ_1+D)^{-1*}(E-Z_1^*Z_2)(CZ_2+D)^{-1},\label{eq:euz}\ee
and, in particular,
\be E-Z'^*Z'=(CZ+D)^{-1*}(E-Z^*Z)(CZ+D)^{-1}.\label{eq:ezz}\ee
Therefore the action of $SU(2,2)$ maps $\cd$ onto $\cd.$ We denote by $\hcd$ the set of all unitary $2\times 2$ matrices - the so called Shilov boundary of $\cd.$  It follows from Eq. (\ref{eq:ezz}) that the transformations of $SU(2,2)$ map $\hcd$ onto itself. $\cd$ is a complex manifold (in fact, it is endowed with a natural K\"{a}hlerian structure), and the transformations of $SU(2,2)$ are holomorphic. By a holomorphic density of weight $n$ we will understand a holomorphic function $\Phi(Z),$ given in each coordinate system $Z,$ with the transformation law:
$ \Phi'(Z')=\det \left(\frac{\partial Z}{\partial Z'}\right)^n\Phi(Z).$ In the following we will need the explicit formula for the (complex) Jacobian determinant for linear fractional transformations.
\begin{lemma}
For transformations of the form (\ref{eq:action}), with $A,B,C,D$ arbitrary $2\times 2$ matrices, we have:
$ \det\left(\frac{\partial Z'}{\partial Z}\right)=\det(M)^2\det(CZ+D)^{-4}.$
provided $CZ+D$ is invertible. In particular, for $M$ in $SU(2,2)$ we have
$ \det\left(\frac{\partial Z'}{\partial Z}\right)=\det(CZ+D)^{-4}.$
\label{lem:det}\end{lemma}
\noindent {\bf Proof}: By differentiation of both sides of Eq. (\ref{eq:action}) we easily get:
$ dZ'=(D-(AZ+B)(CZ+D)^{-1}C)dZ(CZ+D)^{-1},$
where $dZ$ stands for $dZ_{ij},$ $(i,j=1,...,4).$
For a transformation $L$ on $2\times 2$ matrices, of the form $L(X)=AXB,$ we have a general formula \cite{silvester}:
$ \det(L)=\det(A)^2\det(B)^2.$
Writing $B'=AZ+B,$ $D'=CZ+D,$ we thus have:
$ \det\left(\frac{\partial Z'}{\partial Z}\right)=\det(A-B'D'^{-1}C)^2\det(CZ+D)^{-2}.$
It follows then, using Eq. (\ref{eq:dets}) that
$ \det(A-B'D'^{-1}C)\det(D')=\det\left(\begin{smallmatrix}A&B'\\C&D'\end{smallmatrix}\right).$
Now, notice that
$$ \det\left(\begin{smallmatrix}A&B'\\C&D'\end{smallmatrix}\right)=\det\left(\begin{smallmatrix}A&AZ+B\\C&CZ+D\end{smallmatrix}\right)=\det\left(\left(\begin{smallmatrix}A&B\\C&D\end{smallmatrix}\right)\left(\begin{smallmatrix}E&Z\\0&1\end{smallmatrix}\right)\\
\right)=\det(M).
$$Thus the lemma follows. \qed

Let $\cs$ be the set of all coordinate systems obtained from the standard coordinate system of the space of $4\times 4$ complex matrices by $SU(2,2)$ transformations (\ref{eq:action}). Restricting complex densities to these coordinate systems we get, for the transformation rule of a density of weight $n$ the formula
$ \Phi'(Z')=\det(CZ+D)^{4n}\,\Phi(Z).$
In the following we will restrict our attention to coordinate system from $\cs.$ For this class of coordinates we will investigate holomorphic densities of weight $n=1/4,$ with the transformation law
\be \Phi'(Z')=\det(CZ+D)\,\Phi(Z).\label{eq:den14}\ee
We will call them simply {\bf densities}. The vector space of densities will be denoted by $\cf.$
\subsection{Coherent states}
If $\Phi$ is a density, then it is enough to know the function $Z\mapsto \Phi(Z)$ in one coordinate system. It will then be determined in every other coordinate system from $\cs$ using the formula (\ref{eq:action}). Using the standard coordinate system of $\mbox{Mat}(4,\BC)$ to each point $\xi\in\cd$ we will associate a density $\Phi_\xi$ by the following construction: to the origin $\xi=0$ we associate the density $\Phi_0(Z)\equiv 1.$ If $\xi$ is an arbitrary point in $\cd,$ then the matrix $M_\xi$ given by:
\be M_\xi =\left(\begin{smallmatrix}A&B\\C&D\end{smallmatrix}\right),\label{eq:mksi}\ee
with $A=(E-\xi\xi^*)^{-1/2},$ $D=(E-\xi^*\xi)^{-1/2},$ $C=\xi^*A,$ $B=\xi D,$
is easily seen to be in $SU(2,2)$ and it maps $Z=0$ to $Z=\xi.$ The inverse matrix $M_\xi^{-1}$ is then given (cf. Eq. (\ref{eq:inv})) by
\be M_\xi^{-1}=\left(\begin{smallmatrix} (E-\xi\xi^*)^{-1/2}&-(E-\xi\xi^*)^{-1/2}\xi^*\\-(E-\xi^*\xi)^{-1/2}\xi&(E-\xi^*\xi)^{-1/2}\end{smallmatrix}\right).\label{eq:minv}\ee
In the coordinate system $Z'$ obtained from the standard one by the application of $M_\xi^{-1},$  $\xi$ is transformed into $0,$ therefore for the density associate to $\xi$ we should have $\Phi'(Z')\equiv 1.$ It follows then, by using Eqs. (\ref{eq:inv}),(\ref{eq:den14}), that $\Phi_\xi$ should be defined by:
$$ \Phi_\xi(Z)=\det\left((E-\xi^*\xi)^{-1/2}-(E-\xi^*\xi)^{-1/2}\xi^*Z\right)^{-1},$$
or
\be  \Phi_\xi(Z)=\frac{\det \left(E-\xi^*\xi\right)^{1/2}}{\det\left(E-\xi^*Z\right)}.\label{eq:coherent}\ee
We call $\xi\mapsto\Phi_\xi$ the system of {\em coherent states\,}. The system is equivariant in the sense that, for an $SU(2,2)$ transformation
$\xi\mapsto \xi',\quad Z\mapsto Z'$ we have, as can be easily computed, the formula
\be \Phi_{\xi'}(Z')=\frac{\det(C\xi+D)^{*}}{|\det(C\xi+D)|}\,\det(CZ+D)\Phi_\xi(Z).\label{eq:bieq}\ee
The first factor on the right is a pure phase factor. This fact will prove to be of importance later on. The formula (\ref{eq:bieq}) is a particular case of the transformation law of a {\em bi--density\,} of weight $(m,n).$ Denoting by $J(Z)$ the Jacobian determinant of the transformation, we have, for such a bi--density, the formula
\be \Phi(Z_1',Z_2')=\left(\frac{J(Z_1)}{|J(Z_1)|}\right)^{-m}\,J(Z_2)^{-n}\,\Phi(Z_1,Z_2).\label{eq:bidensity}\ee
In our case, with $Z_1=\xi,Z_2=Z),$ we take $m=n=1/4.$

\subsection{The Cayley transform}
The Cayley transform and its inverse are defined as in the $1$--dimensional case by the formula:

$W=i\frac{E-Z}{E+Z},$ $Z=\frac{E+iW}{E-iW}.$
The Cayley transform $w:\,Z\mapsto W$  may be considered as a transformation of the form (\ref{eq:action}) with $A=E,B=E,C=-iE,D=-iE,$ with the determinant of the corresponding matrix $M$ being $\det(M)=-4.$ It follows then from the Lemma \ref{lem:det} that
$ \frac{\partial W}{\partial Z}=16\det(E+Z)^{-4},$
and, using a similar argument,
\be \frac{\partial Z}{\partial W}=16\det(E-iW)^{-4}.\label{eq:dzdw}\ee
\begin{remark} Notice that there is an error in the formula (2.12) of \cite{ruhl}. The corresponding numerical factors there should be $2^{-8}$ and $2^8$ instead of $2^{-4}$ and $2^{12}$ resp. The Jacobian determinants there are for real coordinates, they are squares of absolute values of complex Jacobi determinants as in our formulas above.
\end{remark}
The Cayley transform maps the domain $\cd$ onto the {\em future tube\,} $\ct=\{W=X+iY:\, X=X^*,Y>0\}.$
An open dense subset of the Shilov boundary $\hat{D}$ of $\cd$ is mapped onto the set $\cm$ of all Hermitian $2\times 2$ matrices:
$\cm=\{X\,: X=X^*\}.$
We can use now the formulas (\ref{eq:coherent}),(\ref{eq:bidensity}),(\ref{eq:dzdw}), and obtain the expression of coherent states in terms of the future tube variables $W=w(Z),\, \zeta=w(\zeta)$:
$ \Phi_\zeta(W)=\frac{\det(\zeta-\zeta^*)^{1/2}}{\det(W-\zeta^*)}.$
Let us introduce the standard basis in the space of Hermitian $2\times 2$ matrices $\sigma_0=E,\sigma_1,\sigma_2,\sigma_3$ defined by:
$\sigma_0=\left(\begin{smallmatrix}1&0\cr 0&1\end{smallmatrix}\right),$ $\sigma_1=\left(\begin{smallmatrix}0&1\cr 1&0\end{smallmatrix}\right),$ $\sigma_2=\left(\begin{smallmatrix}0&-i\cr i&0\end{smallmatrix}\right),$ $\sigma_3=\left(\begin{smallmatrix}1&0\cr 0&-1\end{smallmatrix}\right).$
It is convenient to introduce real variables $x^\mu,q^{\mu},l^\nu,\,(\mu=0,...,3) $ via the formulas  \be W=x^\mu\sigma_\mu,\quad \zeta=(q^\mu+il^\mu )\,\sigma_\mu .\ee Notice that we have \be \det(W)=x^2=\eta_{\mu\nu}x^\mu x^\nu,\ee where $\eta_{\mu\nu}=\mbox{diag\,}(+1,-1,-1,-1)$ is the diagonal Minkowski matrix representing the unique (up to a constant scale factor) invariant (with respect to induced action of $SU(2,2)$) conformal structure on $\BR^4.$ Using these new variables the coherent states $\Phi_\zeta(W)$ can be written as: $\Phi_{q,l}(x)=-4l^2\frac{(x-q)^2-l^2-2il(x-q)}{\left(l^2-(x-q)^2\right)^2+4\left(l(x-q)\right)^2}.$
Using the translation we can always make $q=0,$ and then, using a Lorentz rotation, we can get $l=(L,0,0,0).$ This rotationally invariant state reads then as:
$ \Phi_{0,L}(x)=-4R^2\frac{x^2-L^2-2iLx^0}{(L^2-x^2)^2+4L^2(x^0)^2}.$
\subsection{The induced metric}

The Minkowski conformal structure is defined as the constant tensor density $\eta_{\mu\nu}$ of weight $w=-1/2.$ Indeed, if $g_{\mu\nu}$ is a tensor, then $\det(g)$ is a density of weight $2.$ If $\gamma_{\mu\nu}$ is a density of weight $w$, then $\det(\gamma)$ is a density of weight $4+w.$ Therefore, $\det(\gamma)$ can be constant only when $w=-1/2.$ On the other hand the coherent state $\Phi_{0,L}$ is a density of weight $w=1/4.$ It follows that
$ g^L_{\mu\nu}(x)=|\Phi_{0,L}(x)|^2\,\eta_{\mu\nu}$
is a covariant tensor - the space-time metric determined by the coherent state $\Phi_{0,R}.$ Explicitly, written in the standard general relativistic form in radial coordinates $t=x^0,r=\sqrt{(x^1)^2+(x^2)^2+(x^3)^2},\theta,\phi,$ we have
\be ds^2=\frac{16L^4}{L^4+(t^2-r^2)^2+2L^2(t^2+r^2)}\left(dt^2-(dr^2+r^2d\sigma^2)\right),\label{eq:ds2}\ee
where $d\sigma^2=d\theta^2+\sin^2(\theta)d\phi^2.$ After rescaling the $r$ and $t$ coordinates we can, effectively, set $L=1$ to obtain the following conformally flat space--time metric:
$$ ds^2=\frac{1}{1+(t^2-r^2)^2+2(t^2+r^2)}\left(dt^2-(dr^2+r^2d\sigma^2)\right).$$
This is the metric induced by the coherent state $\Phi_\xi$ for $\xi=0.$ The stability group of this point is $S(U(2)\times U(2))$ with the diagonal $U(1)$ subgroup consisting of $SU(2,2)$ matrices of the form:
$ M(\alpha)=\left(\begin{smallmatrix}e^{i\alpha}E&0\\0&e^{-i\alpha}E\end{smallmatrix}\right).$
Via the Cayley transform the action of this subgroup translates into the action on Hermitian $2\times 2$ matrices:
$ W\mapsto \frac{\tan(\alpha)E+W}{E-\tan(\alpha)W}.$
In terms of space--time coordinates $t,r,\theta,\phi$ the trajectories of the action of this $U(1)$ subgroup are $\theta=\mbox{const},\phi=\mbox{const},$ and
\be t(\alpha)=\frac{2\cos(2\alpha)t+\sin(2\alpha)(1+r^2-t^2)}{1+t^2-r^2+\cos(2\alpha)(1+r^2-t^2)-2\sin(2\alpha)t},\ee
\be r(\alpha)=\frac{2r}{1+t^2-r^2+\cos(2\alpha)(1+r^2-t^2)-2\sin(2\alpha)t}.\ee
Differentiating with respect to $\alpha$ at $\alpha=0$ we find the tangent vector field $\Xi$ given by
$ \Xi=(1+r^2+t^2)\frac{\partial}{\partial t}+2rt\frac{\partial}{\partial r}.$
The field $\Xi$ is a {\em radial conformal Killing vector field\,} \cite{herrero2000} for the flat Minkowski metric. It is also, automatically, by its very construction, a true Killing field for the metric given by the line element (\ref{eq:ds2}).
\begin{figure}[ht]
\centering
\includegraphics[width=7cm, keepaspectratio=true]{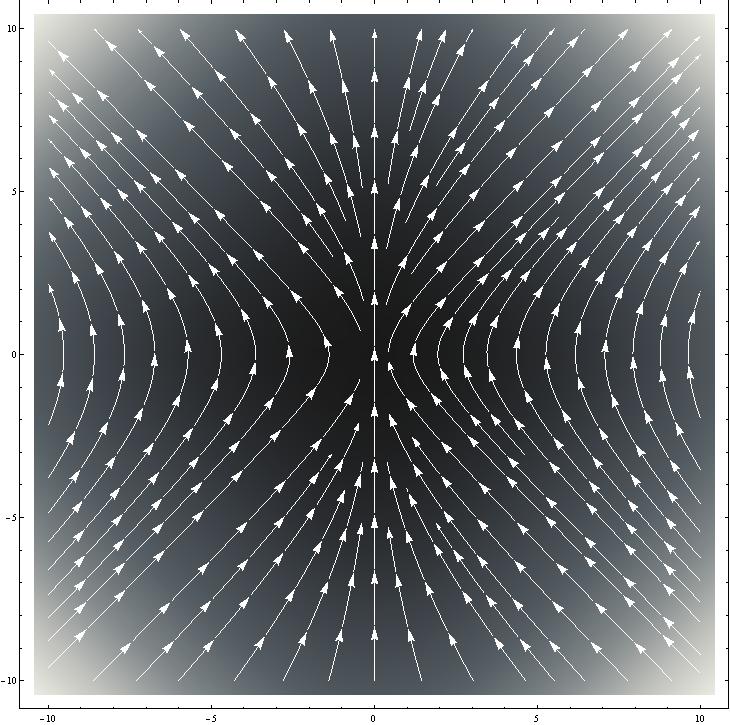}\caption{The Killing vector field $\Xi$ Cayley transformed to the Minkowski space.}
\end{figure}
\subsection{Comoving coordinates}
It is convenient to introduce the comoving coordinates in which the coordinate time is described by the parameter $\alpha$ along the orbits of $\Xi,$ assuming, for instance, that both coordinate systems coincide at $t=0.$ To this end we introduce new coordinates $\tau,\,\rho$ defined by the expressions
$ t=\frac{(1+\rho^2)\sin(2\tau)}{1-\rho^2+(1+\rho^2)\cos(2\tau)},$
$ r=\frac{2\rho}{1-\rho^2+(1+\rho^2)\cos(2\tau)}.$
In new coordinates the line element becomes
$ ds^2=d\tau^2-\frac{1}{(1+\rho^2)^2}(d\rho^2+\rho^2d\sigma^2)$ - the standard form of de Sitter's space--time.

\end{document}